\documentclass[proceedings]{stacs}
\stacsheading{2010}{311-322}{Nancy, France}
\firstpageno{311}

\usepackage{graphicx,psfrag}
\usepackage{epsf}
\usepackage{graphicx,amssymb,amsmath}
\usepackage{epsfig}
\usepackage{amsfonts}
\usepackage{latexsym}
\usepackage{bbm}
\usepackage{wrapfig}

\newcommand{\convh}{{\rm conv}}

\newcommand{\old}[1]{{}}

\def\eps{\varepsilon}
\def\rho{\varrho}

\providecommand{\intd}[0]%
{\;\mbox{d}}

\begin{document}

\title{Long non-crossing configurations in the plane}

\author[lab1]{Adrian Dumitrescu}{Adrian Dumitrescu}
\address[lab1]{Department of Computer Science,
University of Wisconsin--Milwaukee,
WI 53201-0784, USA}
\email{ad@cs.uwm.edu}

\author[lab2]{Csaba D. T\'oth}{Csaba D. T\'oth}
\address[lab2]{Department of Mathematics and Statistics,
University of Calgary, AB, Canada}
\email{cdtoth@ucalgary.ca}

\thanks{Adrian Dumitrescu was supported in part by NSF CAREER grant CCF-0444188.
Part of the research by this author was done at
Ecole Polytechnique F\'ed\'erale de Lausanne.
Csaba D. T\'oth was supported in part by NSERC grant RGPIN 35586. Part
of the research by this author was done at Tufts University.}


\begin{abstract}
We revisit several maximization problems for geometric networks design
under the non-crossing constraint, first studied by Alon, Rajagopalan
and Suri (ACM Symposium on Computational Geometry, 1993).
Given a set of $n$ points in the plane in general position (no three points
collinear), compute a longest non-crossing configuration composed of
straight line segments that is: (a) a matching  (b) a Hamiltonian path
(c) a spanning tree. Here we obtain new results for (b) and (c),
as well as for the Hamiltonian cycle problem:

%\smallskip
(i)  For the longest non-crossing Hamiltonian path problem,
we give an approximation algorithm with ratio $\frac{2}{\pi+1} \approx 0.4829$.
The previous best ratio, due to Alon et al., was $1/\pi \approx 0.3183$.
Moreover, the ratio of our algorithm is close to $2/\pi$ on a relatively broad
class of instances: for point sets whose perimeter (or diameter) is
much shorter than the maximum length matching.
The algorithm runs in $O(n^{7/3}\log{n})$ time.

%\smallskip
(ii) For the longest non-crossing spanning tree problem, we give an
approximation algorithm with ratio $0.502$ which runs in $O(n \log{n})$  time.
The previous ratio, $1/2$, due to Alon et al., was achieved by a
quadratic time algorithm. Along the way, we first re-derive the result
of Alon et al. with a faster $O(n \log{n})$-time algorithm and a very
simple analysis.

%\smallskip
(iii)  For the longest non-crossing Hamiltonian cycle problem,
we give an approximation algorithm whose ratio is close to $2/\pi$ on a
relatively broad class of instances: for point sets with the product
$\bf{\langle}$~diameter~$\times$ ~convex hull size $\bf{\rangle}$ much smaller
than the maximum length matching. The algorithm runs in
$O(n^{7/3}\log{n})$ time. No previous approximation results
were known for this problem.
\end{abstract}

\maketitle

\section{Introduction} \label{sec:intro}

Self-crossing in planar configurations is typically an undesirable
attribute. Many structures studied in computational geometry, in particular
those involving a minimization condition, have the non-crossing
attribute for free, for instance minimum spanning trees, minimum
length matchings, Voronoi diagrams, etc.
The non-crossing property usually follows from the triangle inequality.
Alon et al.~\cite{ARS95} have considered the problems of computing
{\rm (i)} the longest non-crossing matching,
{\rm (ii)} the longest non-crossing Hamiltonian path and
{\rm (iii)} the longest non-crossing spanning tree,
given $n$ points in the plane. Although they were unable to prove it,
they suspected that all these problems are $NP$-hard.
The survey articles by Eppstein~\cite[pp. 439]{Ep00} and Mitchell
\cite[pp. 680]{Mi00} list these as open problems in the area of
geometric network optimization.
The problem of approximating the longest non-crossing
Hamiltonian cycle is also of interest and wide open~\cite[pp. 338]{BE97}.

Without the non-crossing condition explicitly enforced, the problem of
minimizing or maximizing the length of a spanning tree, Hamiltonian
cycle or path, perfect matching, triangulation, etc. has a rich history.
However if such structures are required to be non-crossing much less
is known, in particular for the maximization variants.
While for minimization problems, the non-crossing property comes
usually for free via the triangle inequality, in contrast, for
maximization problems, the non-crossing property conflicts directly
with the length maximizing objective. This is another reason why
these problems are interesting to study.

\medskip
\noindent{\bf Related work.} The existence of non-crossing Hamiltonian
paths and cycles in geometric graphs has been studied in
\cite{ACF+08,CDJK07}.
Various Ramsey-type results for  non-crossing spanning trees, paths
and cycles have been obtained in~\cite{KPT97} and~\cite{KPTV98}.
The Euclidean MAX TSP, the problem of computing a longest
straight-line tour of a set of points, has been proven $NP$-hard in
dimensions three or higher~\cite{Fe99}, while its complexity in the
Euclidean plane remains open~\cite{Mi00}.
In contrast, the shortest non-crossing matching and the
shortest non-crossing spanning tree are both computable in polynomial time
\cite{Ep00,Mi00}, as they coincide with the shortest matching and the
shortest spanning tree respectively.

\medskip
\noindent{\bf Definitions and notations.}
A set $S$ of points in the plane is said to be in {\em  general
position} if no three points are collinear.
General position will be assumed throughout this paper.
Given a set of $n$ points in the plane, the results of Alon  al. are as follows:
{\rm (i)} A non-crossing matching whose total length is at least
$2/\pi$ of the longest (possibly crossing) matching
can be computed in $O(n^{7/3}\log{n})$ time.
{\rm (ii)} A non-crossing Hamiltonian path
whose total length is at least $1/\pi$ of the longest (possibly
crossing) Hamiltonian path can be computed in $O(n^{7/3}\log{n})$ time.
{\rm (iii)} A non-crossing spanning tree whose total length is at least
$n/(2n-2) \geq 1/2$ of the longest (possibly crossing) spanning tree
can be computed in $O(n^2)$ time.
%Their original results mention somewhat higher running times for (i)
%and (ii), which were based on the best upper bound at that time on the
%number of halving lines of a set of $n$ points.
The running times have been
adjusted to reflect the current best upper bound of $O(n^{4/3})$
on the number of halving lines as established by Dey~\cite{De98}.

A {\em geometric graph} $G$ is a pair $(V,E)$ where $V$ is a finite set
of points in general position in the plane, and $E$ is a set set of
segments (edges) connecting points in $V$. The {\em length} of $G$,
denoted $L(G)$, is the sum of the Euclidean lengths of all edges in
$G$. The graph $G$ is said to be {\em non-crossing} if its edges have
pairwise disjoint interiors (collinear triples of points are forbidden
in order to avoid overlapping collinear edges).

For a point set $S$, let $\convh(S)$ be the convex hull of $S$,
and let $P=P(S)$ denote the perimeter of $\convh(S)$.
Denote by $D=D(S)$ the diameter of $S$ and write $n=|S|$. Let
$M_\textrm{OPT}$ be a longest
(possibly crossing) matching of $S$, and let $M^*_\textrm{OPT}$ be a
longest non-crossing matching of $S$; observe that for odd $n$,
$M_\textrm{OPT}$ is a nearly perfect matching, with $(n-1)/2$ edges.
Let $H_\textrm{OPT}$ be a longest (possibly crossing) Hamiltonian path
of $S$, and let $H^*_\textrm{OPT}$ be a longest non-crossing Hamiltonian path
of $S$. Let $T_\textrm{OPT}$ be a longest (possibly crossing) spanning tree
of $S$, and let $T^*_\textrm{OPT}$ be a longest non-crossing spanning tree
of $S$. Finally, let $Q_\textrm{OPT}$ be a longest (possibly crossing)
Hamiltonian cycle of $S$, and let $Q^*_\textrm{OPT}$ be a longest
non-crossing Hamiltonian cycle of $S$. The following inequalities are obvious:
$L(M_\textrm{OPT}) \leq L(H_\textrm{OPT}) \leq L(T_\textrm{OPT})$.

Given a set $S$ of $n$ points in the plane,
a line $\ell$ going through two points of $S$ is called a
{\em halving line} if there are $\lfloor (n-2)/2 \rfloor$ points on one side and
$\lceil (n-2)/2 \rceil$ points on the other side~\cite{L71}.
A {\em bisecting line} $\ell$ of $S$ is any line that partitions the point
set evenly, i.~e., neither of the two open halfplanes defined by
$\ell$ contains more than $n/2$ points of $S$~\cite{E87}.
Observe that any halving line of $S$ is also a bisecting line of $S$.
Any bisecting line of $S$ yields (perhaps non-uniquely) a bipartition
$S=R \cup B$, with $R \cap B =\emptyset$, $||R|-|B|| \leq 1$,
with $R$ contained in one of the closed halfplanes determined by
$\ell$, and $B$ contained in the other.
We call $S=R \cup B$ a linearly separable bipartition, or balanced
partition of $S$. Observe that for any non-zero direction vector
$\vec{v}$, there is a bisecting line orthogonal to $\vec{v}$,
see~\cite[Lemma~4.4]{E87}. Two bisecting lines are called {\em equivalent}
if they can yield the same balanced partition of $S$.
It is well known that the number of non-equivalent bisecting lines
of a set is of the same order as the number of halving lines of the set,
and any balanced bipartition can be obtained from a halving
line~\cite[pp. 67]{E87}.

Our results are summarized in the following three
theorems\footnote{Due to space limitations, some proofs are omitted.}.

\begin{theorem} \label{T1}
{\rm (i) \ } For the longest non-crossing Hamiltonian path problem,
there is an approximation algorithm with ratio $\frac{2}{\pi+1} \approx 0.4829$
that runs in $O(n^{7/3}\log{n})$ time.\\
{\rm (ii) \ }
Given a set of $n$ points in the plane, one can
compute a  non-crossing Hamiltonian path $H$ in $O(n^{7/3}\log{n})$
time such that $L(H) \geq \frac{2}{\pi}L(H_\textrm{OPT})-\frac{P}{\pi}$.
In particular, if the point set satisfies the condition
$\frac{P}{\pi} \leq \delta L(H_\textrm{OPT})$ for some small $\delta>0$,
then $L(H) \geq (\frac{2}{\pi}-\delta)L(H_\textrm{OPT})$.\\
{\rm (iii) \ } Alternatively, one can compute a  non-crossing Hamiltonian path
$H$ in $O(n \log{n}/\sqrt{\eps})$ time,
such that $L(H) \geq (1-\eps)\frac{2}{\pi}L(H_\textrm{OPT})-\frac{P}{\pi}$.
\end{theorem}

\begin{theorem} \label{T2}
For the longest non-crossing spanning tree problem for
a given set of $n$ points in the plane, there is an
approximation algorithm with ratio $0.502$ and $O(n \log{n})$  running
time. More precisely, the algorithm computes a non-crossing spanning
tree $T$ such that %\linebreak
$L(T) \geq 0.502 \cdot L(T_\textrm{OPT})$.
\end{theorem}

\old{
Although our improvement in the approximation ratio for spanning trees
is very small, it shows that the ``barrier'' of $1/2$ can be broken.
Also, while from a practical standpoint the improvement
in the running time is the most significant aspect, from a theoretical
perspective the improvement in the approximation ratio is the most
challenging part of our result.
} % old

\begin{theorem} \label{T3}
Given a set $S$ of $n$ points in the plane, with $|\convh(S)|=h$: \\
{\rm (i) \ } One can
compute a  non-crossing Hamiltonian cycle $Q$ in $O(n^{7/3}\log{n})$
time such that $L(Q) \geq \frac{2}{\pi}L(Q_\textrm{OPT})-(2h-1) \frac{P}{\pi}$.
In particular, if the point set satisfies the condition
$(2h-1) \frac{P}{\pi} \leq \delta L(Q_\textrm{OPT})$
%(respectively, $(h+2)D \leq \delta L(Q_\textrm{OPT})$)
for some small $\delta>0$,
then $L(Q) \geq \left(\frac{2}{\pi}-\delta\right) L(Q_\textrm{OPT})$. \\
{\rm (ii)\ } Alternatively,  one can compute a non-crossing
Hamiltonian cycle $Q$ in  $O(n^3 \log{n})$ time such that $L(Q) \geq
\frac{2}{\pi}L(Q_\textrm{OPT})- (h+2) \frac{P}{\pi}$. \\
{\rm (iii)\ } Alternatively, one can compute a  non-crossing
Hamiltonian cycle $Q$ in $O(n \log{n}/\sqrt{\eps})$ time, such that
$L(Q) \geq (1-\eps)\frac{2}{\pi}L(Q_\textrm{OPT})-(2h-1) \frac{P}{\pi}$.
\end{theorem}

\section{The Hamiltonian path} \label{sec:path}

In this section we prove Theorem~\ref{T1}. Let $S=\{p_1,\ldots,p_n\}$.
We follow an approach similar to that of Alon et al. using projections and
an averaging argument, in conjunction with a result on bipartite
embeddings of spanning paths in the plane. Abellanas et
al.~\cite[Theorem~3.1]{AGH+99} showed that every linearly separable bipartition
$S=R \cup B$ with $||R|-|B|| \leq 1$, admits an alternating non-crossing
spanning path such that the edges cross any separating line $\ell$ at points
ordered monotonically along $\ell$. Such a Hamiltonian path can be computed
in $O(n \log{n})$ time. Their algorithm computes the same Hamiltonian path for
any two equivalent halving lines, that is, the alternating path depends on the
bipartition only rather than the separating line.

We now recall the algorithm of Abellanas et al.~\cite{AGH+99};
see Fig.~\ref{abell} for an example.
Let $S=R \cup B$ with $||R|-|B|| \leq 1$ be the red-blue bipartition
given by a vertical line $\ell$: $R$ on the left, $B$ on the right.
Their algorithm constructs an alternating path $A$ in the following way:
Let $rb$ be the top red-blue edge of
the convex hull $\convh(S)$, called the {\em top bridge}. If $|R|>|B|$,
set $A := \{r\}$, if $|R|<|B|$, set $A := \{b\}$, else
set $A$ to $\{r\}$ or $\{b\}$ arbitrarily. At every step, recompute
the top bridge $rb$ of $S \setminus A$, and add $r$ to $A$
if the last point in $A$ was blue, or add $b$ to $A$
if the last point in $A$ was red.  As pointed out by the authors, the
resulting path $A$ is non-crossing because $A$ is disjoint from the
convex hull of $S \setminus A$ at each step.

We improve the lower bound of Alon et al. by computing the longest
Hamiltonian path corresponding to a bipartition and a Hamiltonian path of
length at least the perimeter of the convex hull, and returning the
longest of the two.
\begin{lemma} \label{L1}
For a point set $S$, $|S| =n \geq 31$, a non-crossing Hamiltonian
path of length at least $P(S)$ can be computed in $O(n \log{n})$
time. The bound on the length is best possible.
\end{lemma}
%
%\begin{proof}
%\end{proof}

Consider a geometric graph $G=(V,E)$, and a point $q \notin V$,
so that $V \cup \{q\}$ is in general position.
We say that $q$ {\em sees a vertex} $v \in V$ if the segment $q v$  does not intersect any
edge of $G$. Similarly, we say that $q$ {\em sees an edge} $e \in E$,
if the triangle formed by $v$ and $e$ does not intersect any other edge of
$G$.
We make use of the fact that if $n$ is even then the two endpoints of an
alternating path are on opposite sides of the separating line $\ell$.
If $n$ is odd, we first construct an alternating path for a specific subset of
$n-1$ points, and then augment it to a Hamiltonian path on all $n$ points using
the following lemma.

\begin{lemma} \label{L2}
Let $S=R \cup B$ with $||R|-|B|| \leq 1$, be a linearly separable bipartition
given by line $\ell$. Let $q \in S$, and $A'$ be a non-crossing alternating
path on $S \setminus \{q\}$ such that its $($consecutive$)$ edges cross $\ell$ at
points ordered monotonically along $\ell$. Then $q$ sees one edge of
$A'$ and consequently, $A'$ can be extended to a Hamiltonian path $A$ on $S$, with
$L(A')<L(A)$. The path $A$ can be computed in $O(n)$ time, given $A'$.
\end{lemma}
%
%\begin{proof}
%\end{proof}

Fix a Cartesian coordinate system $\Gamma$.
Let $k$ be the number of halving lines of $S$, denote the angles they
make with the $x$-axis of $\Gamma$ by $0 \leq \alpha_1 < \ldots \alpha_k < \pi$.
By relabeling the points assume that the optimal path is
$H_\textrm{OPT} = p_1,p_2,\ldots,p_n$.
For two points $p_i, p_j \in S$, let $\beta_{ij}$ be the angle in $[0,\pi)$
formed by the line through $p_i p_j$ and the $x$-axis.
If $n$ is odd, then a bisecting line of direction $\alpha$ (for any $\alpha$)
must be incident to at least one point of $S$, and denote an arbitrary such
point by $q_\alpha$.

\smallskip
\noindent Algorithm {\bf A1}:\\
{\sc Step 1}. Compute a non-crossing
Hamiltonian path $H_1$ of length at least $P(S)$, by Lemma~\ref{L1}.\\
{\sc Step 2}. If $n$ is even, then for all non-equivalent bisections of
$S$ (i.e., for all balanced bipartitions of $S$), compute a non-crossing
alternating path using the algorithm of Abellanas et al.~\cite{AGH+99}, and
let the longest such path be $H_2$. If $n$ is odd, then for all non-equivalent
bisections of $S$, compute a non-crossing
alternating path of the even point set $S\setminus \{q_\alpha\}$ using the
algorithm of~\cite{AGH+99} and let the longest such path be $H_2'$.
Augment $H_2'$ with vertex $q_\alpha$ by Lemma~\ref{L2} to a Hamiltonian path
$H_2$.\\
{\sc Step 3}. Output the longest of the two paths $H_1$ and $H_2$.

\smallskip
By Lemma~\ref{L1}, the running time of {\sc Step 1} is
$O(n \log{n})$. Since the number of halving lines of an $n$-element
point set is $O(n^{4/3})$ and all can be generated within this time~\cite{De98},
the running time of {\sc Step 2} is $O(n^{7/3}\log{n})$, consequently
the total running time of {\bf A1} is also $O(n^{7/3}\log{n})$.
%\smallskip

We proceed with the analysis of the approximation ratio.
For simplicity, we assume that $n$ is even. The case of $n$ odd is
slightly different.
For each $\alpha \in [0,\pi)$, let $\Gamma_\alpha$
be a (rotated) coordinate system, obtained from $\Gamma$ via a
counterclockwise rotation by $\alpha$, and with the $y$-axis dividing
evenly the point set $S$. Let $x_i$ be the $x$-coordinate of point
$p_i$ with respect to $\Gamma_\alpha$.
For a given $\alpha$, let $H_\alpha$ be a non-crossing alternating
path with respect to a balanced bipartition induced by the $y$-axis of
$\Gamma_\alpha$, as computed by the algorithm.
There are $O(1)$ balanced bipartitions given by any halving line of
$S$. Recall that $H_\alpha$ does not depend continuously on
$\alpha$; it depends only on the discrete bipartition. However, the
coordinates of the points depend continuously on $\alpha$.
Assume that $H_\alpha = p_{\sigma(1)},p_{\sigma(2)},\ldots,p_{\sigma(n)}$,
where $\sigma$ is a permutation of $[n]$; here $\sigma$ depends
on the bipartition (hence also on $\alpha$).
Let $W_\alpha$ denote the {\em width} of $S$ in direction $\alpha$, that is,
the width of the smallest parallel strip of direction $\alpha$ that contains $S$.
By projecting on the $x$-axis of $\Gamma_\alpha$, we get
\begin{eqnarray} \label{eq:111}
L(H_\alpha) &\geq& |x_{\sigma(1)}| + 2|x_{\sigma(2)}| +
\ldots+ 2|x_{\sigma(n-1)}| + |x_{\sigma(n)}|
=2\sum_{i=1}^n |x_i| - |x_{\sigma(1)}| - |x_{\sigma(n)}| \nonumber\\
&=& \sum_{j=1}^{n-1} (|x_j| +|x_{j+1}|)+
|x_1|+ |x_n|  - |x_{\sigma(1)}| - |x_{\sigma(n)}|
\geq \sum_{j=1}^{n-1} (|x_{j}| +|x_{j+1}|) -W_\alpha \nonumber\\
&\geq&  \sum_{j=1}^{n-1} |p_{j} p_{j+1}| |\cos(\beta_{j j+1}-\alpha)| -W_\alpha
\end{eqnarray}
In the 2nd line of the above chain of inequalities, we use the fact
that $p_{\sigma(1)}$ and $p_{\sigma(n)}$  lie on opposite sides of $\ell$,
since $n$ is even, hence $ |x_{\sigma(1)}| +|x_{\sigma(n)}|
\leq |p_{\sigma(1)} p_{\sigma(n)}| \leq W_\alpha$,
In the 3rd line, we make use of the following inequality:
for any two points $p_i, p_j \in S$,
$|p_i p_j| |\cos (\beta_{ij}-\alpha) | \leq |x_i| + |x_j|$,
with equality if and only if the two points lie on opposite sides of
the $y$-axis of $\Gamma_\alpha$.
Recall: for even $n$, $H_2$ is the longest of the $O(k)$ Hamiltonian
non-crossing paths $H_{\alpha_i}$ over all $O(k)$ balanced
bipartitions of $S$. (A given angle $\alpha_i$ yields $O(1)$
balanced partitions, and corresponding alternating paths denoted here
$H_{\alpha_i}$.) We thus have for each $\alpha \in [0, \pi)$:
$$L(H_2) \geq \sum_{j=1}^{n-1} |p_{j} p_{j+1}|
|\cos(\beta_{j j+1}-\alpha)| -W_\alpha.$$
Note that
$$ \int_{0}^{\pi} |\cos(\beta_{j j+1}-\alpha)| \intd \alpha =
\int_{0}^{\pi} |\cos \alpha| \intd \alpha =2,$$
and according to Cauchy's surface area formula,
we have $\int_{0}^{\pi} W_\alpha \intd \alpha = P(S)$.
By integrating both sides of the previous inequality over the
$\alpha$-interval $[0,\pi]$, we obtain
$$ \pi L(H_2)
\geq 2 \sum_{j=1}^{n-1} |p_{j} p_{j+1}| -P(S)
= 2 L(H_\textrm{OPT}) -P(S), $$
\begin{equation} \label{E1}
L(H_2) \geq \frac{2}{\pi}L(H_\textrm{OPT})-\frac{P(S)}{\pi}.
\end{equation}
We now improve the old approximation ratio of $\frac{1}{\pi} \approx
0.3183$ to $\frac{2}{\pi+1} \approx 0.4829$, by
balancing the lengths of the two paths computed in {\sc Step 1} and
{\sc Step 2}. Set $c=\frac{\pi+1}{2}$.

\smallskip
{\em Case 1:} $L(H_\textrm{OPT}) \leq c P(S)$. By
considering the path computed in {\sc Step 1}, we get a ratio of
at least
$$ \frac{L(H_1)}{L(H_\textrm{OPT})} \geq
\frac{P(S)}{L(H_\textrm{OPT})} \geq \frac{P(S)}{cP(S)}= \frac{2}{\pi+1}. $$

\smallskip
{\em Case 2:} $L(H_\textrm{OPT}) \geq c P(S)$. By considering the path
computed in {\sc Step 2} (inequality \eqref{E1}), we get a ratio of at least
$$ \frac{L(H_2)}{L(H_\textrm{OPT})} \geq
\frac{\frac{2}{\pi} L(H_\textrm{OPT})-\frac{1}{\pi}P(S)}{L(H_\textrm{OPT})} \geq
\frac{2}{\pi} - \frac{1}{c\pi} = \frac{2}{\pi} \left(1- \frac{1}{\pi+1}\right) =
\frac{2}{\pi+1}. $$

\smallskip
Observe that if the point set satisfies the condition
$\frac{P(S)}{\pi} \leq \delta L(H_\textrm{OPT})$, then by \eqref{E1},
we have
$$
L(H) \geq \frac{2}{\pi} L(H_\textrm{OPT})- \delta L(H_\textrm{OPT}) =
\left(\frac{2}{\pi}-\delta \right) L(H_\textrm{OPT}).
$$
This concludes the proofs of parts (i) and (ii) of Theorem~\ref{T1}.

\smallskip
(iii) With the same approach as in~\cite{ARS95}, a Hamiltonian path of
length at least $(1-\eps)\frac{2}{\pi}L(H_\textrm{OPT})- \frac{P(S)}{\pi}$
can be found by considering only $b/\sqrt{\eps}$ angles
$\theta_i=\frac{i \pi \sqrt{\eps}}{b}$, for
$i=0,1,\ldots,\lfloor b/\sqrt{\eps} \rfloor$, where $b$ is a suitable
absolute constant. The resulting running time is $O(n \log{n}/\sqrt{\eps})$.
This concludes the proof of Theorem~\ref{T1}.

\section{The spanning tree} \label{sec:tree}

In this section we prove Theorem~\ref{T2}.
Let $S=\{p_1,\ldots,p_n\}$, where $p_i=(x_i,y_i)$.
Given a point $p\in S$, the {\em star centered at} $p$,
denoted $S_p$, is the spanning tree on $S$ whose edges join $p$ to all
the other points. Since $S$ is in general position, $S_p$ is
non-crossing for any $p \in S$.
An {\em extended star centered at} $p$ is a spanning tree of $S$
consisting of paths of length $1$ or $2$ (edges) connecting $p$ to all
the other points.
See Fig.~\ref{f2}. While the star centered at a point is unique, there
may be many extended stars centered at the same point, and some of
them may be self-crossing. In particular $S_p$ is also an extended
star.

\begin{figure} [htbp]
\centerline{\epsfxsize=3.3in \epsffile{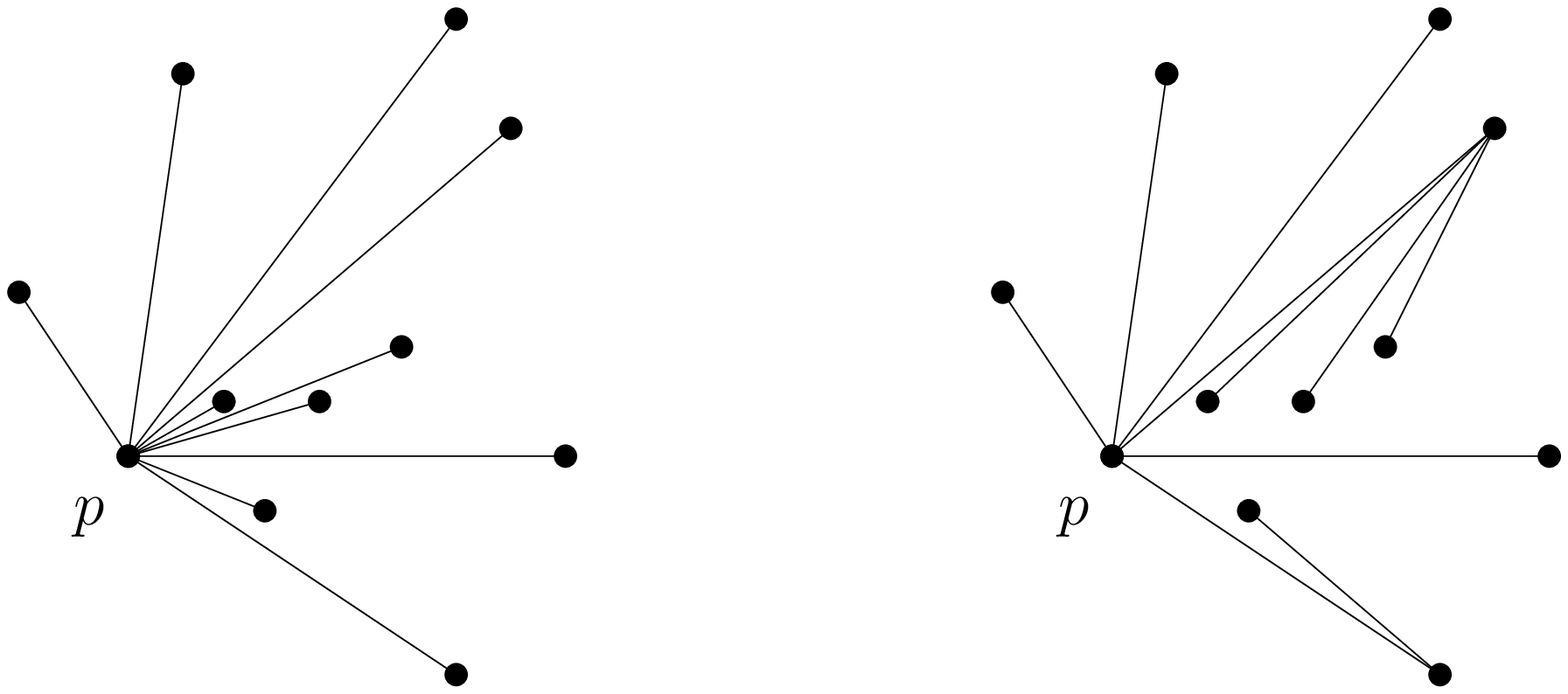}}
\caption{\small A star (left) and a non-crossing extended star (right)
on a same point set, both centered at the same point $p$.}
\label{f2}
\end{figure}

The algorithm of Alon et al. computes the $n$ stars centered at each
of the points, and then outputs the longest one. The algorithm takes
quadratic time, and the analysis shows a ratio of $\frac{n}{2n-2}$
(which tends to $1/2$ in the limit). Their algorithm works in any
metric space. As pointed out by Alon et al., the ratio $1/2$ is best
possible (in the limit) for this specific algorithm.
We first re-establish the $1/2$ approximation ratio using a faster
algorithm, and also with a simpler analysis. Our algorithm
works also in any  metric space; however in this general setting, the
running time remains quadratic.

\smallskip
\noindent Algorithm {\bf A2}: Compute a diameter of the point
set, and output the longest of the two stars centered at one of its endpoints.

\smallskip
Obviously the algorithm runs in $O(n \log{n})$ time, with bottleneck
being the diameter computation~\cite{PS85}.
Let $ab$ be a diameter pair, and assume w.l.o.g. that $|ab|=1$.
The ratio $1/2$ (or even $\frac{n}{2n-2}$)
follows from the next lemma in conjunction with the obvious upper
bound $L(T_\textrm{OPT}) \leq n$ (or $L(T_\textrm{OPT}) \leq n-1$).

\begin{lemma} \label{L3}
Let $S_a$ and  $S_b$ be the stars centered at the points $a$ and $b$,
respectively. Then $ L(S_a) + L(S_b) \geq n$.
\end{lemma}
\begin{proof}
Assume that $a=p_1$, $b=p_2$. For each $i=3,\ldots,n$, the triangle inequality
for the triple $a,b,p_i$ gives
$$ |ap_i| + |b p_i| \geq |ab|=1. $$
By summing up we have
\[ L(S_a) + L(S_b)= \sum_{i=3}^n (|ap_i| + |b p_i|) + 2|ab|
\geq (n-2)+2=n. \]

\vspace{-1.5\baselineskip}
\end{proof}

We now continue with the new algorithm that achieves a (provable)
$ \frac{1}{2} + \frac{1}{500}$ approximation ratio within the same
running time $O(n \log{n})$.

\smallskip
\noindent Algorithm {\bf A3}: Compute a diameter $ab$ of the point
set, and output the longest of the 5 non-crossing structures
$S_a$, $S_b$, $S_h$, $E_a$, $E_b$, described below.

\smallskip
Assume w.l.o.g. that the $ab$ is a horizontal unit segment, where
$a=(0,0)$ and $b=(1,0)$.
Let $h=(x_h,y_h)$ be a point in $S$ with a largest value of $|y|$.
By symmetry, we can assume that $y_h \geq 0$.
$S_a$, $S_b$, and $S_h$ are the 3 stars centered at $a$, $b$, and
$h$ respectively. $E_a$, resp. $E_b$, are two non-crossing
extended stars centered at $a$, resp, $b$; details to follow.
Each of the five structures can be computed in $O(n \log{n})$
time, so the total execution time is also $O(n \log{n})$.

\smallskip
Set $\delta=0.05$, $w=0.6$, $t=0.6$ and $z=0.48$, and refer to
Fig.~\ref{f1}.  Let $\ell_1$, $\ell_2$, $\ell_3$, and $\ell_4$,
be four parallel vertical lines:
$\ell_1: x=0$,  $\ell_2: x=0.2$, $\ell_3: x=0.8$, $\ell_4: x=1$.
Obviously, all points in $S$ lie in the strip bounded by $\ell_1$ and
$\ell_4$.
Let $V_m$ be the vertical parallel strip symmetric about the midpoint
of $ab$ and of width $w$. We refer to $V_m$ as the middle strip;
$V_m$ is bounded by the vertical lines $\ell_2$ and $\ell_3$.
Let $V_a$ and $V_b$ be the two vertical strips of width $0.2$
bounded by $\ell_1$ and $\ell_2$, and by $\ell_3$ and $\ell_4$
respectively. Let $c=(x_c,y_c)$ be the intersection point between
$\ell_3$ and the circular arc $\gamma_a$ of unit radius centered at
$a$ and sub-tending an angle of $60^\circ$. We have $x_c=0.8$ and
$$ y_c = \sqrt{1-0.8^2} = 0.6= t. $$

\begin{figure} [htb]
\centerline{\epsfxsize=2.7in \epsffile{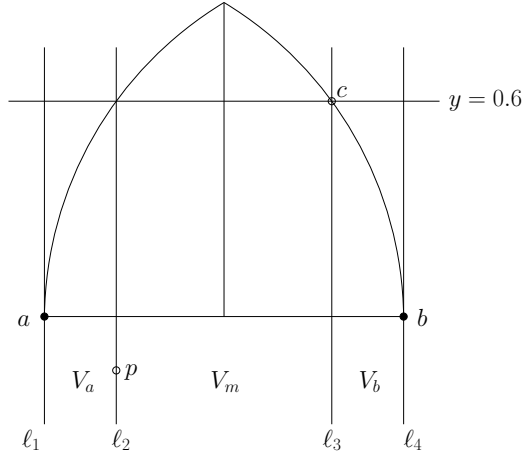}}
\caption{\small A diameter pair $a,b$ at unit distance, and the three
vertical strips $V_a$, $V_m$, and $V_b$. The two circular arcs
$\gamma_a$  and $\gamma_b$ of unit radius centered at $a$ and $b$
intersect at the point $(1/2, \sqrt{3}/2)$. All points of $S$ above
$ab$ lie in the region bounded by $ab$, $\gamma_a$  and $\gamma_b$.}
\label{f1}
\end{figure}
We now describe the two extended star structures $E_a$ and $E_b$.
See also Fig.~\ref{f3} for an example.
To construct $E_a$, first compute the order of visibility of the
points in $V_b$ from point $a$ by sorting. Then connect $a$ with each
point in the right strip $V_b$. Note that $b \in V_b$, thus $V_b \neq
\emptyset$. Call $S'_a$ the resulting star. The edges of this star together
with the vertical line $\ell_3$ divide $V_a \cup V_m$
into convex regions (wedges with a common apex $a$) ordered top-down.
The subset of points in each wedge can be computed using binary search
in overall $O(n \log{n})$ time (over all wedges).
$S'_a$ is extended (augmented) as follows.
In each wedge, say $paq$, all points are connected either to
$a$ or to $p$, depending on the best (longest) overall connection
cost. We denote the resulting {\em extended star} structure by $E_a$.
The construction of $E_b$ is analogous. It is clear by construction
that both $E_a$ and $E_b$ are non-crossing.

\begin{lemma} \label{L4}
For each $p \in S$, let $d_\textrm{max}(p)$ denote the maximum distance
from $p$ to other points in $S$. Then
$$ L(T_\textrm{OPT}) \leq \left[\sum_{i=1}^n d_\textrm{max}(p_i)\right] -1. $$
\end{lemma}
\begin{proof}
Consider $T_\textrm{OPT}$ rooted at $a$ and drawn as an abstract tree
with the root at the top in the usual manner. Let $\pi(v)$ denote the
parent of a (non-root) vertex $v$. Uniquely assign each
edge $\pi(v) v$ of $T_\textrm{OPT}$ to vertex $v$. Obviously,
$L(\pi(v) v) \leq d_\textrm{max}(v)$ holds for each edge in the tree.
By adding up the above inequalities, and taking into account that
$d_\textrm{max}(a)=|ab|=1$, the lemma follows.
\end{proof}

\begin{lemma} \label{L5}
Assume that $\sum_{i=1}^n |y_i| \geq \delta n$ for some positive constant
$\delta \leq 1$. Then
$$ L(S_a) + L(S_b) \geq 2n \sqrt{\frac{1}{4} + \delta^2}. $$
\end{lemma}

\begin{lemma} \label{L6}
Let $n_a$ and $n_b$ denote the number of points in the left and right
vertical strips $V_a$ and $V_b$. Then $L(E_a) \geq \frac{1+w}{4} (n+n_b)$,
and similarly $L(E_b) \geq \frac{1+w}{4} (n+n_a)$. Consequently
$ L(E_a) + L(E_b) \geq \frac{1+w}{4} (2n+n_a+n_b)$.
$E_a$ and $E_b$ can be constructed in $O(n \log{n})$ time.
\end{lemma}
\begin{proof}
The distance between $\ell_1$ and $\ell_3$ is $\frac{1+w}{2}$.
By an argument similar to that in the proof of Lemma~\ref{L3},
the connection cost for a wedge with $m$ points is
at least $\frac{1+w}{4}m$. Therefore the total length of $E_a$ is
$$ L(E_a) \geq \frac{1+w}{2}n_b + \frac{1+w}{4}(n-n_b)=
\frac{1+w}{4} (n+n_b). $$
The estimation of $L(E_b)$ is analogous. The running time has been
established previously.
\end{proof}

\begin{lemma} \label{L7}
Assume that $\sum_{i=1}^n |y_i| \leq \delta n$ and $y_h \geq t$.
Then $L(S_h) \geq (t-\delta)n$.
\end{lemma}
\begin{proof}
\[ L(S_h) \geq  \sum_{i=1}^n (y_h - y_i)= n y_h - \sum_{i=1}^n  y_i
\geq n y_h - \sum_{i=1}^n  |y_i| \geq n y_h - \delta n
\geq (t-\delta)n. \]

\vspace{-1.5\baselineskip}
\end{proof}

\begin{lemma} \label{L8}
Assume that $|y_h| \leq t=0.6$.
Let $p \in S$ be a point in the middle strip $V_m$, with $y$-coordinate
satisfying $|y| \leq 0.15$. Then $d_\textrm{max}(p) \leq 0.9605$.
\end{lemma}
\begin{proof}
It is straightforward to check that the maximum distance is attained
for a point $p$ on $\ell_2$ with $y$-coordinate $-0.15$.
The furthest point from $p$ in the allowed region is $c$. Hence
\[ d_\textrm{max}(p) \leq |pc| = \sqrt{w^2 + (0.15+ t)^2} =
\sqrt{0.6^2 + 0.75^2} \leq 0.9605. \]

\vspace{-1.5\baselineskip}
\end{proof}

We now distinguish the following four cases to complete our
estimation of the approximation ratio.

\smallskip
{\em Case 1:} $\sum_{i=1}^n |y_i| \geq \delta n$. The algorithm
outputs\footnote{Here and in other instances it is meant that the
algorithm outputs a structure at least as long as these.}
$S_a$ or $S_b$. By Lemma~\ref{L5}, the approximation ratio is
at least

$$ \frac{L(S_a)+L(S_b)}{2 L(T_\textrm{OPT})} \geq \sqrt{\frac{1}{4} + \delta^2}
\geq 0.502. $$

\smallskip
{\em Case 2:} $\sum_{i=1}^n |y_i| \leq \delta n$ and
$y_h \geq t$. The algorithm outputs $S_h$.
By Lemma~\ref{L7}, the approximation ratio is
at least $t-\delta =0.55$.

\smallskip
{\em Case 3:} $\sum_{i=1}^n |y_i| \leq \delta n$ and
$y_h \leq t$ and $n_a+n_b \geq (1-z)n$.
The algorithm outputs $E_a$ or $E_b$.
We only need the last inequality in estimating the length.
By Lemma~\ref{L6}, the approximation ratio is at least

$$ \frac{L(E_a)+L(E_b)}{2 L(T_\textrm{OPT})} \geq
\frac{1+w}{4}  \cdot \frac{2n+n_a+n_b}{2n} \geq \frac{(1+w)(3-z)}{8}
= \frac{1.6 \cdot 2.52}{8}= 0.504. $$

\smallskip
{\em Case 4:} $\sum_{i=1}^n |y_i| \leq \delta n$ and
$y_h \leq t$ and $n_a+n_b \leq (1-z)n$.
The algorithm outputs $S_a$ or $S_b$.
There are at least $z n =0.48 n$ points in the middle strip $V_m$.
Observe that at most $n/3$ points in $V_m$ have
$|y_i| \geq 0.15$; otherwise we would have
$$ \sum_{i=1}^n |y_i| \geq \sum_{V_m} |y_i| > 0.15 \cdot \frac{n}{3}
= 0.05 n =\delta n, $$
a contradiction. It follows that at least $12n/25-n/3=11n/75$
points in the middle strip have $|y_i| \leq 0.15$. By Lemma~\ref{L4}
and Lemma~\ref{L8},
$$ L(T_\textrm{OPT}) \leq \frac{64n}{75} + 0.9605 \cdot \frac{11n}{75}
\leq 0.9943 n. $$
The approximation ratio is at least
$$ \frac{L(S_a)+L(S_b)}{2 L(T_\textrm{OPT})} \geq \frac{n}{2 \cdot 0.9943 n}
\geq 0.502. $$
This completes the list of cases and thereby the proof of Theorem~\ref{T2}.

\smallskip
\noindent {\em Remark.}
The example in Fig.~\ref{f3} with $n$ points ($n$ even)
equally spaced along a circle shows that the constant $0.502$
measuring the approximation ratio achieved by our algorithm {\bf A3}
cannot be improved to anything larger than $2/\pi$.
Indeed the lengths of the five structures computed by the algorithm
are $L(S_a)=L(S_b)=L(S_h)=L(E_a)=L(E_b) =(1-o(1)) \frac{2}{\pi} n$,
while $L(T_\textrm{OPT}) \geq L(H_\textrm{OPT}) = (1-o(1))n$.
\begin{figure} [htb]
\centerline{\epsfxsize=2.1in \epsffile{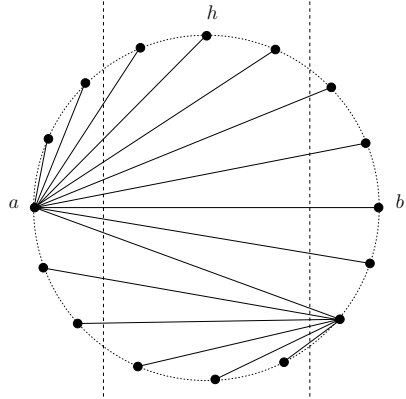}}
\caption{\small The non-crossing structure $E_a$ for an example with
$n=16$ points on the circle. The middle strip $V_m$ is bounded by the
two dashed vertical lines.}
\label{f3}
\end{figure}

\section{The Hamiltonian cycle} \label{sec:cycle}

In this section we present the proof of Theorem~\ref{T3}, which is similar
(including notation) to that of Theorem~\ref{T1}.
The rotated coordinate system $\Gamma_\alpha$, and the $x$-coordinates
$x_i$ with respect to this system are denoted in the same way.
By relabeling the points assume that the optimal cycle is
$Q_\textrm{OPT} = p_1,p_2,\ldots,p_n$ (with the convention that $p_{n+1}=p_1$).
We approximate $Q_{\rm OPT}$ by constructing a non-crossing alternating path $A$
on a subset of $S$, and then completing it to a non-crossing cycle using
convex hull vertices. We need to observe that the alternating path $A$ on the
subset $I$ of interior (non-hull) vertices of $S$ produced by the algorithm of
Abellanas et al.~\cite{AGH+99} is {\em not} good enough for {\em this} strategy:
even though one endpoint of $A$ (the first computed by the algorithm) is
always on the convex hull of $I$, the other endpoint might be blocked by edges
of $A$, so that $A$ might not be extendible to a non-crossing Hamiltonian cycle
(an example is shown in Fig.~\ref{abell}). Here, we give a stronger result
that fits our purpose (for an even number of points).
\begin{figure} [htb]
\centerline{\epsfxsize=2in \epsffile{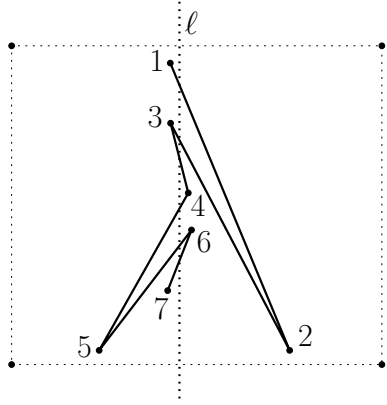}}
\caption{\small A non-crossing alternating path obtained by the algorithm of
  Abellanas et al. For the purpose of cycle construction, the path
is non-extendible from its 2nd endpoint, vertex~$7$.}
\label{abell}
\end{figure}

\begin{lemma}\label{L9}
Let $S=R \cup B$ with with $|R|=|B|$, be a linearly separable bipartition
given by line $\ell$. Then $S$ admits an alternating non-crossing
spanning path $A$ such that
(1) the edges of $A$ cross $\ell$ at points ordered monotonically
along $\ell$; and (2) the two endpoints of $A$ are incident to the two
distinct edges of the convex hull that connect $R$ and $B$ (the two
red-blue bridges). Such a Hamiltonian path can be computed in $O(n \log{n})$ time.
We refer to the underlying procedure as the {\em two-endpoint path
construction algorithm}.
\end{lemma}
\begin{proof}
We modify the algorithm of Abellanas et al. for path construction, so
that the path is grown from the two endpoints and the two sub-paths merge ''in
the middle''.
Recall that  $S=R \cup B$, and $|R|=|B|$, thus $|S|$ is even.
Let $r_1 b_1$ and $r_2 b_2$ be the top and bottom red-blue edges of
the convex hull $\convh(S)$, respectively, called {\em top} and {\em bottom}
bridges; it is possible that $r_1=r_2$ or  $b_1=b_2$ but not both.
One endpoint of $A$ is an endpoint of the top bridge, and the other
endpoint of $A$ is an endpoint of the bottom bridge, and they are
chosen of opposite colors.
% Since the convex hull of $S$ may be a triangle (and the top and bottom
% bridge might be adjacent), we assume that $|R|=|B|$.
Let $A=\{r_1,b_2\}$ or $A=\{b_1,r_2\}$ arbitrarily,
containing two endpoints of the path. At every step, recompute
the top and bottom bridges of $S \setminus A$, and append either
the red or the blue vertex of each bridge to $A$ such that the appended
edges cross the separating line $\ell$. In the last step, the convex hull of
$S\setminus A$ is a red-blue segment that merges the two sub-paths.
The two new edges added simultaneously at each step
cannot cross each other; and they cannot cross previous edges, since they are
separated from them by the convex hull of $S\setminus A$. Finally, they cannot
extend the two sub-paths by the same point either, because $|S|$ is even.
\end{proof}

The next lemma follows from~\cite[Lemma~2.1]{HKRT08}; we will only need its
corollary, Lemma~\ref{L11}.

\begin{lemma} {\rm(\cite{HKRT08}).} \label{L10}
Let $P=p_1, p_2, \ldots, p_n$ be a simple polygon (with the convention
that $p_{n+1}=p_1$) and $q$ be a point in the exterior of the convex
hull of $P$, where $P \cup \{q\}$  is in general position. Then $q$
sees one edge $p_i p_{i+1}$ of $P$. Such an edge can be found in $O(n)$ time.
\end{lemma}
\begin{lemma} \label{L11}
Let $P=p_1, p_2, \ldots, p_n$ be a simple polygon (with the convention
that $p_{n+1}=p_1$) and $q$ be a point
in the exterior of the convex hull of $P$, where $P \cup \{q\}$
is in general position. Then the polygonal cycle
$P$ can be extended to include $q$ so that $P \cup \{q\}$ is
still a simple polygon. More precisely, there exists $i \in [n]$, so that
$Q= p_1, \ldots, p_i, q, p_{i+1}, \ldots, p_n$ is a simple polygon.
Moreover, $L(Q)>L(P)$. The extension can be computed in $O(n)$ time.
\end{lemma}
\begin{proof}
By Lemma~\ref{L10}, $q$ sees one edge $p_i p_{i+1}$ of $P$.
Replacing this edge of $P$ by the two edges $p_i q $ and $q p_{i+1}$
results in a simple polygon $Q= p_1, \ldots, p_i, q, p_{i+1}, \ldots, p_n$.
By the triangle inequality, $L(Q)>L(P)$.
The extension can be computed in $O(n)$ time, as determined by the
time needed to find a visible edge.
\end{proof}

\smallskip
Note that the condition in the lemma that $q$ lies in the exterior of
the convex hull of $P$, is indeed necessary. Otherwise one cannot
guarantee that $q$ sees an edge of $P$.

\smallskip
{\rm (i)} Let $S=S' \cup S''$, where $S'$ is the set of convex hull
vertices and $S''$ is the set of interior points.
Let $S'=\{p_{j_1}, p_{j_2}, \ldots, p_{j_h}\}$.
Put $h =|S'|$, $m =|S''|$, thus $n=h+m$. Assume first for simplicity that
$m$ is even. An easy modification of the algorithm, explained below, is used
for $m$ odd.

\smallskip
\noindent Algorithm {\bf A4}: \\
{\sc Step 1.} For all non-equivalent
bisections of  $S''$ (i.e., for all balanced bipartitions of $S''$):
1.~Compute a non-crossing alternating path $A$ by using the two-endpoint path
construction algorithm (Lemma~\ref{L9}).
2.~Extend $A$ to a cycle by connecting its endpoints to (one or two)
convex hull vertices. 3.~Further extend this cycle to include the
remaining hull vertices, by repeated invocation of Lemma~\ref{L11}. \\
{\sc Step 2.} Output the longest such cycle (containing all points of $S$).

\smallskip
Observe that after {\sc Step 1.}1, the two endpoints of the path are vertices
of $\convh(S'')$, hence they can be connected to hull vertices to make a cycle.
If $m$ is odd, then there is a point $q \in S''$ on the line $\ell$.
Use the two-endpoint path construction algorithm for $S''\setminus\{q\}$, and
the same bisecting line $\ell$. If $q$ is in the interior of
$\convh(S''\setminus \{q\})$, then extend the path with point $q$, using
Lemma~\ref{L2}. Otherwise, $q$ sees the top or bottom bridge of
$\convh(S''\setminus \{q\})$, so the path can be extended by connecting $q$
to the endpoint visible to $q$. The two endpoints of the extended path are on
$\convh(S'')$, hence they can be connected to hull vertices to make a
cycle, as in the case of even $m$.


\begin{thebibliography}{10} % \itemsep -2pt

\bibitem {AGH+99} M. Abellanas, J. Garcia, G. Hern\'andez,
M. Noy, and P. Ramos:
Bipartite embeddings of trees in the plane,
{\em Discrete Applied Mathematics}, {\bf 93} (1999), 141--148.

\bibitem{ACF+08}
O.~Aichholzer, S.~Cabello, R.~Fabila-Monroy, D.~Flores-Pe\~naloza,
T.~Hackl, C.~Huemer, F.~Hurtado, and D.~R.~Wood:
Edge-removal and non-crossing configurations in geometric graphs,
{\em Proc. 24th European Workshop on Computational Geometry}, Nancy 2008,
pp.~119--122.

\bibitem {ARS95} N. Alon, S. Rajagopalan and S. Suri:
Long non-crossing configurations in the plane,
{\em Fundamenta Informaticae} {\bf 22} (1995), 385--394.
Also in {\em Proc. 9th ACM Sympos. on Comput. Geom.},
1993, 257--263.

%\bibitem {AGS00} T. Asano, S. Ghosh, and T. Shermer:
%Visibility in the plane,
%in {\em Handbook of Computational Geometry}
%(J.-R. Sack and J.~Urrutia, editors),
%Elsevier Science, Amsterdam, 2000, pp.~829--876.

%\bibitem {Ba96} A. Barvinok:
%Two algorithmic results for the traveling salesman problem,
%{\em Mathematics of Operations Research}, {\bf 21} (1996), 65--84.

\bibitem {BE97} M. Bern and D. Eppstein:
Approximation algorithms for geometric problems,
in {\em Approximation Algorithms for $NP$-hard Problems}
(D. S. Hochbaum, editor), PWS, Boston, 1997, pp.~296--345.

%\bibitem {BMP05} P.~Bra\ss , W.~Moser, and J.~Pach:
%{\em Research Problems in Discrete Geometry},
%Springer, New York, 2005.

\bibitem{CDJK07}
J. \v{C}ern\'y, Z.~Dvo\'r\'ak, V.~Jel\'{\i}nek, and J.~K\'ara:
Noncrossing {H}amiltonian paths in geometric graphs,
{\em Discrete Applied Mathematics}, {\bf 155} (2007), 1096--1105.

\bibitem {De98} T. K. Dey:
Improved bounds on planar $k$-sets and related problems,
{\em Discrete \& Computational Geometry}, {\bf 19} (1998), 373--382.

\bibitem {E87} H. Edelsbrunner:
{\em Algorithms in Combinatorial Geometry},
Springer-Verlag, Heidelberg, 1987.

%\bibitem {EW85} H. Edelsbrunner and E. Welzl:
%On the number of separations of a finite set in the plane,
%{\em Journal of Combinatorial Theory, Ser. A}, {\bf 38} (1985), 15--29.

\bibitem {Ep00} D. Eppstein:
Spanning trees and spanners,
in {\em Handbook of Computational Geometry}
(J.-R. Sack and J.~Urrutia, editors),
Elsevier Science, Amsterdam, 2000, pp.~425--461.

%\bibitem {EA81} H. ElGindy and D. Avis:
%A linear algorithm for computing the visibility polygon from a point,
%{\em Journal of Algorithms}, {\bf 2} (1981), 186--197.

%\bibitem{ELSS73}
%P.~Erd\H os,  L. Lov\'asz, A. Simmons, and E. Straus:
%Dissection graphs of planar point sets. In
%{\em A Survey of Combinatorial Theory} (J. N. Srivastava, editor),
%North-Holland, Amsterdam, pp.~139--154, 1973.

\bibitem {Fe99} S. P. Fekete:
Simplicity and hardness of the maximum traveling salesman problem
under geometric distances, 
{\em Proceedings of the 10th ACM-SIAM Symposium on Discrete Algorithms},
1999, pp.~337--345.

\bibitem{HKRT08}
F.~Hurtado, M.~Kano, D.~Rappaport, and Cs.~D.~T\'oth:
Encompassing colored planar straight line graphs,
{\em Computational Geometry: Theory and Applications}, {\bf 39} (1) (2008), 14--23.

%\bibitem {J90} B. Joe:
%On the correctness of a linear-time visibility polygon algorithm,
%{\em Internat. Journal Comput. Math.}, {\bf 32} (1990),
%155--172.

%\bibitem {JS87} B. Joe and R. B. Simpson:
%Correctness to Lee's visibility polygon algorithm,
%{\em BIT}, {\bf 27} (1987), 458--473.

\bibitem {KPT97} G. K\'arolyi, J. Pach and G. T\'oth:
{Ramsey-type results for geometric graphs. I},
{\it Discrete and Computational Geometry} {\bf 18} (1997), 247-255.

\bibitem {KPTV98} G. K\'arolyi, J. Pach, G. T\'oth and P. Valtr:
{Ramsey-type results for geometric graphs. II},
{\it Discrete and Computational Geometry}, {\bf 20} (1998), 375--388.

%\bibitem {Lee89} D. T. Lee:
%Visibility of a simple polygon,
%{\em Comput. Vision, Graphics and Image Process.}, {\bf 22} (1983),
%207--221.

\bibitem {L71} L. Lov\'asz:
On the number of halving lines,
{\it Ann. Univ. Sci. Budapest, E\"otv\"os, Sec. Math.},
{\bf 14} (1971), 107--108.

\bibitem {Mi00} J.~S.~B. Mitchell:
Geometric shortest paths and network optimization,
in {\em Handbook of Computational Geometry}
(J.-R. Sack and J.~Urrutia, editors),
Elsevier Science, Amsterdam, 2000, pp.~633--701.

%\bibitem {OR04} J. O'Rourke: Visibility,
%in {\em Handbook of Discrete and Computational Geometry}
%(J.~Goodman and J. O'Rourke, editors),
%Chapman \& Hall, 2nd edition, 2004, pp.~643--663.

\bibitem {PS85} F. Preparata and M. Shamos:
{\em Computational Geometry: An Introduction},
Springer, New York, 1985.

%\bibitem{P57} R.~C.~Prim: Shortest connection networks and some
%generalizations, {\em Bell System Technical Journal} {\bf 36}
%(1957), 1389--1401.

%\bibitem {T01} G. T\'oth:
%Point sets with many $k$-sets,
%{\it Discrete \& Computational Geometry},
%{\bf 26} (2001), 187--194.

%\bibitem{YB61}
%I.~M.~Yaglom and V.~G.~Boltyanski:
%{\em Convex Figures}, Holt, Rinehart and Winston,
%New York, 1961.

\end{thebibliography}
\end{document}